\documentclass[12pt,reqno]{amsart}

\newcommand\version{October 4, 2010}

\usepackage{amsmath,amsfonts,amsthm,amssymb,amsxtra}

\newtheorem{theorem}{Theorem}[section]

\newtheorem{lemma}[theorem]{Lemma}
\newtheorem{corollary}[theorem]{Corollary}

\theoremstyle{definition}

\theoremstyle{remark}

\newtheorem{remark}[theorem]{Remark}

\numberwithin{equation}{section}

\newcommand{\const}{\mathrm{const}\ }

\renewcommand{\epsilon}{\varepsilon}

\newcommand{\kk}{\mathbf{k}}

\newcommand{\p}{\mathbf{p}}
\renewcommand{\phi}{\varphi}
\newcommand{\R}{\mathbb{R}}

\newcommand{\x}{\mathbf{x}}
\newcommand{\y}{\mathbf{y}}

\DeclareMathOperator{\Tr}{Tr}

\begin{document}

\title[Multi-polaron Systems --- \version]{Binding, Stability, and Non-binding\\ of Multi-polaron Systems}

\thanks{\copyright\, 2010 by the authors. This paper may be reproduced, in its entirety, for non-commercial purposes.}

\author[R. L. Frank]{Rupert L. Frank}
\address{Department of Mathematics,
Princeton University, Princeton, NJ 08544, USA}
\email{rlfrank@math.princeton.edu}

\author[E. H. Lieb]{Elliott H. Lieb}
\address{Departments of Mathematics and Physics, Princeton University, P.~O.~Box 708, Princeton, NJ 08542, USA}
\email{lieb@princeton.edu}

\author[R. Seiringer]{Robert Seiringer}
\address{Department of Mathematics, McGill University, 805 Sherbrooke Street West, Montreal, QC H3A 2K6, Canada}
\email{rseiring@math.mcgill.ca}

\author[L. E. Thomas]{Lawrence E. Thomas}
\address{Department of Mathematics, University of Virginia, Charlottesville, VA 22904, USA}
\email{let@virginia.edu}

\begin{abstract}
  The binding of polarons, or its absence, is an old and subtle topic. After defining the model we state some recent theorems of ours. First, the transition from many-body collapse to the existence of a thermodynamic limit for $N$ polarons occurs precisely at $U=2\alpha$, where $U$ is the electronic Coulomb repulsion and $\alpha$ is the polaron coupling constant. Second, if $U$ is large enough, there is no multi-polaron binding of any kind. We also discuss the Pekar-Tomasevich approximation to the ground state energy, which is valid for large $\alpha$. Finally, we derive exact results, not reported before, about the one-dimensional toy model introduced by E. P. Gross.
\end{abstract}


\maketitle

\section{Definition and previous rigorous results}

The large polaron, first considered by H. Fr\"ohlich \cite{Fr} in 1937, is a model of an electron moving in three dimensions and interacting with the quantized optical modes of a polar crystal. (It is called `large' because the size of the electronic wave function is large compared to the crystal lattice spacing, so a continuum approximation is appropriate.) It is also a simple quantum field theory model and over the years has been used as a testing ground for various approximations.

In suitable units, its Hamiltonian is
\begin{align}\label{eq:h1}
 H^{(1)} =\p^2 & + \int_{\R^3} a^\dagger(\kk)a(\kk)\,d\kk  + \frac{\sqrt{\alpha}}{\sqrt 2\, \pi} \int_{\R^3} \frac{1}{|\kk|}[a(\kk)\exp(i\kk\cdot\x) +h.c.] \,d\kk \,.
\end{align}
This Hamiltonian acts in the Hilbert space $L^2(\R^3)\otimes \mathcal F$, where $\mathcal F$
is the bosonic Fock space for the longitudinal optical modes of the
crystal, with scalar creation and annihilation operators $a^\dagger(\kk)$
and $a(\kk)$ satisfying $[a(\kk),a^\dagger(\kk')]=\delta(\kk-\kk')$.
The momentum of an electron is $\p=-i\nabla$, and the coupling constant is $\alpha>0$. (Other authors have used a different convention, where $\alpha$ is replaced by $\alpha/\sqrt2$ \cite{Fr,DePeVe}.) 

The ground state energy $E^{(1)}(\alpha)$ is the infimum of the spectrum of $H^{(1)}$.
Because of translation invariance,  $E^{(1)}(\alpha)$ cannot be expected to be an
eigenvalue,  and indeed it is not \cite{GeLo,Fr2}. The following rigorous results concerning $E^{(1)}(\alpha)$ will be important
in our analysis.
\begin{enumerate}
 \item[(i)] For all $\alpha$, 
$$
-\alpha-\alpha^2/3\leq E^{(1)}(\alpha)\leq -\alpha \,.
$$
These upper and lower bounds are in \cite{Gu,LePi,LeLoPi} and \cite{LiYa}, respectively. As a consequence, $E^{(1)}(\alpha)\sim -\alpha$ for $\alpha$ small.
\item[(ii)] Pekar \cite{Pe}, using a product function, showed that 
$$
E^{(1)}(\alpha)\leq - C_P \alpha^2 \,,
$$
for all $\alpha$. The lower bound
$$
E^{(1)}(\alpha)\geq - C_P \alpha^2 - \const \alpha^{9/5}
$$
for large $\alpha$ was proved in \cite{LiTh}. It was proved earlier in \cite{DoVa}, but without the $\alpha^{9/5}$-error estimate. Here, $C_P=0.109$ \cite{Mi} is the number determined by Pekar's variational problem for the electron density,
\begin{align}\label{eq:pekar}
C_P = \inf\left\{ \int_{\R^{3}} |\nabla \psi|^2 \,d\x - \iint_{\R^3\times\R^3} \frac{|\psi(\x)|^2\, |\psi(\y)|^2}{|\x-\y|} \,d\x\,d\y : \|\psi\|_2 =1  \right\} \,.
\end{align}
The minimizing $\psi$ is unique up to translations (and a trivial phase) \cite{Li}.
\item[(iii)] There is a representation for $E^{(1)}(\alpha)$ in
terms of path integrals.  In terms of the partition function
$Z^{(1)}(T)=\Tr \exp\big(-T H^{(1)}\big)$, one has $E^{(1)}(\alpha) =
-\lim_{T\to\infty}T^{-1}\log Z^{(1)}(T)$.  (Strictly speaking,
$Z^{(1)}(T)$ does not exist because of the translation invariance of
$H^{(1)}_0$ and the infinite number of phonon modes. These
technicalities can be handled by inserting appropriate cutoffs, to be
removed at the end of the calculation \cite{Ro,Sp}.) It was shown in
\cite{Fe} that after one integrates out the phonon variables,
$Z^{(1)}(T)$ has a functional integral representation
\begin{equation}
 \label{eq:trace}
Z^{(1)}(T) = \int d\mu \exp\left[ \frac\alpha2 \int_0^T\!\! \int_0^T \!\frac{e^{-|t-s|}\,dt\,ds}{|\x(t)-\x(s)|} \right] \,,
\end{equation}
where $d\mu$ is Wiener measure on all $T$-periodic paths $\x(t)$. (In physics notation $d\mu=\exp(-\int_0^T \dot\x(t)^2 \,dt)\, d\,\mathrm{path}$. Strictly speaking, $t-s$ has to be understood modulo $T$, but this is irrelevant as $T\to \infty$.)
\end{enumerate}


\section{The multi-polaron problem and new results}

Of great physical interest is the binding energy of $N$ polarons, with Hamiltonian
\begin{align}\label{eq:ham}
 H^{(N)}_U = & \sum_{j=1}^N p_{j}^2 + \int a^\dagger(\kk)a(\kk)\,d\kk + \frac{\sqrt{\alpha}}{\sqrt2\,\pi} \sum_{j=1}^N\int \frac{1}{|\kk|}[a(\kk)\exp(i\kk\cdot\x_j)+h.c.]\,d\kk \notag \\
& +U \sum_{1\leq i<j\leq N} |\x_i -\x_j|^{-1}
\end{align}
and ground state energy $E^{(N)}_U(\alpha)$. We ignore Fermi statistics for the electrons, because its imposition changes things only quantitatively, not qualitatively. The Coulomb repulsion parameter $U\geq 0$ is equal to $e^2$. The derivation of $H^{(N)}_U$ in \cite{Fr} implies that $U>2\alpha$,  and this is crucial for thermodynamic stability, as we shall see.

The generalization of \eqref{eq:trace} is
\begin{equation}
 \label{eq:traceN}
Z^{(N)}_U(T) = \!\!\int d\mu^{(N)} \exp\left[ \frac\alpha2 \sum_{i,j} \int_0^T\!\! \int_0^T \!\frac{e^{-|t-s|}\,dt\,ds}{|\x_i(t)-\x_j(s)|} 
- U \sum_{i<j} \int_0^T \!\frac{dt}{|\x_i(t) -\x_j(t)|} \right] \,,
\end{equation}
where $d\mu^{(N)}$ is Wiener measure on all $T$-periodic paths $(\x_1(t),\ldots,\x_N(t))$. This is relevant for us since $E^{(N)}_U(\alpha) = -\lim_{T\to\infty}T^{-1}\log Z^{(N)}_U(T)$.

The generalization of the Pekar approximation \eqref{eq:pekar} to the $N$-electron case is the minimization of the following Pekar-Tomasevich functional for normalized functions $\psi$ on $\R^{3N}$,
\begin{align*}\nonumber
\sum_{i=1}^N \int_{\R^{3N}} |\nabla_i \psi|^2 \,dX + U \sum_{i<j} \int_{\R^{3N}} \frac{|\psi(X)|^2}{|\x_i-\x_j|} \,dX 
- \alpha \iint_{\R^3\times\R^3} \frac{\rho_\psi(\x)\, \rho_\psi(\y)}{|\x-\y|} \,d\x\,d\y \,,
\end{align*}
where $dX=\prod_{k=1}^N d\x_k$. The density $\rho_\psi$ of $\psi$ is defined as usual by
\begin{equation*}
\rho_\psi(\x) = \sum_{i=1}^N \int_{\R^{3(N-1)}} |\psi(\x_1,\ldots,\x,\ldots,\x_N)|^2 \,d\x_1\cdots \widehat{d\x_i} \cdots d\x_N
\end{equation*}
with $\x$ at the $i$-th position, and $\widehat{d\x_i}$ meaning that $d\x_i$ has to be omitted in the product $\prod_{k=1}^N d\x_k$. This minimization problem is obtained from the original problem of minimizing $\langle\Psi,H^{(N)}_U\Psi\rangle$ by restricting the allowed $\psi$'s to be of the form $\psi\otimes \Phi$, where $\psi\in L^2(\R^{3N})$, $\Phi\in\mathcal F$, and both $\psi$ and $\Phi$ are normalized. Since the Pekar-Tomasevich functional is the result of a variational calculation, its energy gives an upper bound to the ground state energy $E^{(N)}_U(\alpha)$.

\subsection{Binding of multi-polaron systems}

We first consider the bipolaron binding energy $\Delta E_U(\alpha)= 2 E^{(1)}(\alpha)-E^{(2)}_U(\alpha)$.
For some time this was thought to be zero for all $U\geq 2\alpha$, on the basis of an 
inadequate variational calculation, but it is now known \cite{DePeVe}
to be positive for some $U>2\alpha$. The first question we address is
whether $\Delta E_U(\alpha)=0$ for $U$ sufficiently large.
It is understood that the effective interaction induced by the phonon
field for two polarons at large distances $d$ is approximately
Coulomb-like $-2\alpha/d$, but this alone does not preclude binding. (The reason for $2\alpha\cdot \mathrm{distance}^{-1}$ can be seen from the Wiener integral representation \eqref{eq:traceN}, where there is a factor $\alpha/2$, but the pair $(i,j)$ appears twice, and the integral $\int_\R e^{-|t-s|}\,ds = 2$.)
The known existence of bipolarons for some $U>2\alpha$ is an effect of
correlations. It is \emph{a priori} conceivable that correlations lead to an
effective attraction that is stronger than Coulomb at large distances.
If it were, for example, equal to $(2\alpha/d)\log(\log(\log(d)))$,
then this minuscule perturbation of Coulomb's law, which would be
virtually undetectable by a variational calculation, would result in
binding for \emph{all} $U$. The absence of binding is a problem that
has resisted a definitive resolution for many years. The following was proved in \cite{FrLiSeTh0,FrLiSeTh}:

\begin{theorem}[\textbf{Absence of binding for bipolarons}]\label{thm2}
 Let $N=2$. For some constant  $C<26.6$,
\begin{equation}
 \label{eq:bipol}
E^{(2)}_U(\alpha) = 2 E^{(1)}(\alpha)
\end{equation}
whenever $U\geq 2C\alpha$.
\end{theorem}

The constant $26.6$ vastly exceeds the current, non-rigorous estimates of about $1.15$ \cite{VSPD,fomin}, so it is an \emph{open problem} to find a more accurate rigorous bound.

The existence of a critical repulsion strength for a bipolaron is consistent with the idea that the attractive interaction induced by the field is Coulomb-like, and therefore one expects that there is an $N$-independent $U_c(\alpha)$ such that there is no binding of any kind when $U>U_c(\alpha)$. This was proved in \cite{FrLiSeTh0,FrLiSeTh} as well.

\begin{theorem}[\textbf{Absence of binding for {\mathversion{bold}$N$} polarons}]\label{thm:nobindingN}
 For given $\alpha>0$ there is a finite $U_c(\alpha)>2\alpha$ such that
\begin{equation}
 \label{eq:bindabs}
E^{(N)}_U(\alpha) = N E^{(1)}(\alpha)
\quad\text{for all}\ N\geq 2
\end{equation}
whenever $U\geq U_c(\alpha)$.
\end{theorem}

\begin{remark}
 If $U> U_c(\alpha)$, then given \eqref{eq:bindabs} and any normalized $\psi$
\begin{equation}
 \label{eq:energybound}
\left\langle \psi \left| H_U^{(N)} \right|\psi\right\rangle \geq N E^{(1)}(\alpha) + (U-U_c(\alpha)) \ \Big\langle \psi \Big|  \sum_{i<j} |\x_i-\x_j|^{-1} \Big|\psi \Big\rangle \,.
\end{equation}
This inequality gives a quantitative estimate of the energy penalty needed to bring two or more particles within a finite distance of each other. In particular, it implies that for $U>U_c(\alpha)$ there cannot be a normalizable ground state, even in a fixed momentum sector. Inequality \eqref{eq:energybound} is not only true for our bound on $U_c(\alpha)$, but also for the  (unknown) exact value of the critical repulsion parameter.
\end{remark}

We state Theorem \ref{thm2} separately for two reasons: One is that the proof is easier than for the general $N$ case. The second is that we have an upper bound on $U_c(\alpha)$ that is linear in $\alpha$. While our $N$-polaron bound is linear in $\alpha$ for large $\alpha$, we have not achieved this linear bound for small $\alpha$ and this remains an \textit{open problem}.

\subsection{Thermodynamic stability of multi-polaron systems}

The second problem we consider is the existence of the thermodynamic
limit. For large $N$, physical intuition suggests that $E^{(N)}_U(\alpha)\sim
-\const N$. This supposition is known to be false if $U<2\alpha$.
Indeed, it was shown in \cite{GrMo} that, even with the Pauli
principle, $E^{(N)}_U(\alpha)\sim -\const N^{7/3}$ when $U<2\alpha$. Absent
the Pauli principle, $E^{(N)}_U(\alpha)$ would behave even worse, as $-\const
N^3$. It is also known \cite{GrMo} that $E^{(N)}_U(\alpha)\geq -\const N^{2}$
if $U>2\alpha$.  The latter bound ought to be $-\const N$ instead, and in \cite{FrLiSeTh0,FrLiSeTh} we proved that this is indeed the case for all $U> 2\alpha$.

\begin{theorem}[\textbf{Thermodynamic stability for {\mathversion{bold} $U>2\alpha$}}] \label{thm:stabalpha}
 For given $U>2\alpha>0$, $N^{-1} E^{(N)}_U(\alpha)$ is bounded independently of $N$.
\end{theorem}

Our lower bound on $N^{-1} E^{(N)}_U(\alpha)$ goes to $-\infty$ as
$U\searrow2\alpha$, but we are not claiming that this reflects the
true state of affairs. Whether $\lim_{N\to\infty} N^{-1}
E^{(N)}_{2\alpha}(\alpha)$ is finite or not remains an \textit{open
  problem}. There are partial results in the Pekar-Tomasevich approximation \cite{GrMo}.

The linear lower bound from Theorem \ref{thm:stabalpha}, together with the sub-additivity of the energy \cite{GrMo}, \cite[Sec. 14.2]{LiSe}, i.e., 
\begin{equation}\label{eq:subadd}
E^{(N+M)}_U(\alpha) \leq E^{(N)}_U(\alpha)+ E^{(M)}_U(\alpha)\,,
\end{equation}
implies:

\begin{corollary}[\textbf{Thermodynamic limit for {\mathversion{bold} $U>2\alpha$}}] \label{thm:tlalpha}
 For given $U>2\alpha>0$, $\lim_{N\to\infty} N^{-1} E^{(N)}_U(\alpha)$ exists.
\end{corollary}

For $U$ in the range $2\alpha<U<U_c(\alpha)$, there are bound states
of an undetermined nature. Does the system become a gas of bipolarons, or does it coalesce into a true $N$-particle bound state? If the latter, does this state exhibit a periodic structure, thereby forming a super-crystal on top of the underlying lattice of atoms? This is perhaps the physically most interesting \textit{open problem}. While particle statistics does not play any role for our main results, the answer to this question will crucially depend on particle statistics (Bose or Fermi) \cite{SVPD1,SVPD2}.


\section{Absence of bipolaron binding}

In order to give the flavor of our methods, we sketch the proof of Theorem \ref{thm2}, as given in \cite{FrLiSeTh0}, \copyright Amer. Phys. Soc. The proofs of the other two theorems are also sketched in \cite{FrLiSeTh0}.

 The proof of Theorem 1 is conveniently structured in 4 steps.

\emph{Step 1. Partition of the interparticle distance:} We fix a
length $\ell$, whose value will later be chosen proportional to
$\alpha^{-1}$, and partition the relative distance $r=|\x_1-\x_2|$
between the particles into spherical shell-like regions of radial size
$2^{k-1}\ell\leq r\leq 2^k\ell$ with $k=1,2,\ldots$. This partitioning
is one of the key points of our analysis. In addition there is the
$k=0$ region, where the particle separation is between zero and
$\ell$. Because of the uncertainty principle these regions have to
overlap a bit, but this can be easily handled, and we ignore it for the sake of
simplicity. There is a kinetic energy cost for localizing the
particles according to this partition, which is $c_1 2^{-2k} \ell^{-2}$
in the shell $k$. In the next step we look at the energy of the
particles localized to one of these shell-like regions.

\emph{Step 2. Further localization for well-separated particles:} For
$k\geq 1$ we further localize the particles into individual boxes of
size $2^{k-3}\ell$. This costs another localization error $c_2 2^{-2k} \ell^{-2}$. Because the separation exceeds $2^{k-1}\ell$, the two particles cannot be in the same or neighboring boxes.  From the path
integral \eqref{eq:traceN}, but now with the $\x_i(t)$'s constrained to
their respective boxes, we see that \emph{the separated particles feel
  an effective Coulomb-like attractive potential}. However, this can
contribute at worst $-c_3\alpha 2^{-k} \ell^{-1}$ to the energy. But
the Coulomb repulsion is at least $U 2^{-k}\ell^{-1}$, which implies
that the total energy exceeds $2E^{(1)}$ if
\begin{equation*}
 \label{eq:klarge}
U 2^{-k}\ell^{-1} > c_3\alpha 2^{-k} \ell^{-1} + (c_1+c_2) 2^{-2k} \ell^{-2} \,.
\end{equation*}
If this inequality holds for $k=1$, it holds for all $k\geq 2$ as well. Thus, if we 
can deal with the $k=0$ region, we will establish that binding is not possible if 
\begin{equation}
\label{eq:klarge1}
 U\alpha^{-1} > c_3 + (c_1+c_2)/(2\ell\alpha) \,.
\end{equation}

\emph{Step 3. The region of no minimal separation:}
In the $k=0$ region, the Coulomb repulsion is at least $U \ell^{-1}$,  but since there is no minimal separation, we have no direct handle on the possible attraction due to the field. We need a lemma, which we will prove in Step 4. It concerns $E^{(2)}_0(\alpha)$, the energy of the bipolaron with no Coulomb repulsion, i.e., $U=0$;
\begin{equation}
 \label{eq:binding}
E^{(2)}_0(\alpha) \geq 2E^{(1)}(\alpha) - 7\alpha^2/3
\quad \text{for all}\ \alpha \,.
\end{equation}
Assuming this, the total energy in the $k=0$ region exceeds $2E^{(1)}(\alpha)$ provided
\begin{equation*}
 U \ell^{-1} > 7\alpha^2/3 + c_1 \ell^{-2} \,,
\end{equation*}
that is, no binding occurs if
\begin{equation}
\label{eq:ksmall1}
 U\alpha^{-1} > 7\ell\alpha/3 + c_1 /(\ell\alpha) \,.
\end{equation}
Setting the right sides of \eqref{eq:klarge1} and \eqref{eq:ksmall1} equal leads to
 the choice $\ell=c_4/\alpha$ and to absence of binding if $U>C\alpha$, as asserted.

\emph{Step 4. The universal lower bound \eqref{eq:binding}:}
In this step, $U=0$. We first note that 
\begin{equation*}
 \label{eq:twiceenergy}
E^{(1)}(2\alpha)\geq 2 E^{(1)}(\alpha)-4\alpha^2/3 \,.
\end{equation*}
This follows from the lower bound $E^{(1)}(\alpha)\geq -\alpha -\alpha^2/3$ in \cite{LiYa} and the upper bound $E^{(1)}(\alpha)\leq -\alpha$ in \cite{Gu,LePi,LeLoPi}, stated above. So \eqref{eq:binding} will follow if we can prove that
 \begin{equation}
 \label{eq:twoenergy}
E^{(2)}_0(\alpha)\geq E^{(1)}(2\alpha)-\alpha^2 \,.
\end{equation}
For this purpose we go back to the functional integral
\eqref{eq:traceN} and use Schwarz's inequality $\langle e^{a+b} \rangle
\leq \langle e^{2a} \rangle^{1/2} \langle e^{2b} \rangle^{1/2}$, where
$\langle\cdot\rangle$ now denotes expectation with respect to Wiener
measure.  We choose $a$ to be the sum of the two terms $i=j=1$ and
$i=j=2$ in \eqref{eq:trace}, and $b$ to be the mixed terms $i\neq j$.
Since $\langle e^{2a} \rangle^{1/2} \sim e^{-T E^{(1)}(2\alpha)}$ for
large $T$, inequality \eqref{eq:twoenergy} will be achieved if we can
show that $\langle e^{2b} \rangle^{1/2}\sim e^{T\alpha^2}$.  At first
sight, the double path integral $\langle e^{2b} \rangle$ looks like
that for a positronium-like atom, i.e., two particles attracting each
other through a Coulomb force with coupling constant $4\alpha$. The
trouble is that the interaction in \eqref{eq:trace} is at different
times, i.e., $|\x_1(t)-\x_2(s)|^{-1}$. A simple application of Jensen's
inequality, however, shows that we can fix the time difference $u=t-s$
and obtain the bound
$$
\langle e^{2b} \rangle
\!\leq \!\!\int_{-\infty}^\infty \!\!\frac{e^{-|u|} du}2 \!\int\! d\mu^{(2)} \!\exp\!\left[ 4\alpha \!\!\int_0^T \!\!\!\frac{\,dt}{|\x_1(t)-\x_2(t-u)|} \right]
$$
Because of the $T$-periodic time translation invariance of the Wiener
measure, the path integral is, in fact, independent of $u$. Hence we
get the positronium-like answer as a bound.  This completes our
argument for the universal bound (\ref{eq:binding}) and hence the absence of bipolaron
binding for sufficiently large $U/\alpha$.


\section{Polarons in one dimension}

In 1976 E. P. Gross \cite{Gr} wrote a seminal paper on the polaron in which he discussed a one-dimensional version. Even though it is not very physical, this model has been widely studied \cite{Sp,SmKoVePeDe,VaPeSmDe} and we are able to prove an interesting theorem about it which we report here for the first time. While we have ignored the Fermi statistics up to now, it will play an important role in this section.

There are $N$ particles on the real line at $x_1,\ldots,x_N\in\R$. We assume they are fermions, but with $q$ spin states for each particle. The case $q=N$ is equivalent to saying that Fermi statistics is irrelevant, i.e., one is dealing with boltzons. The Hamiltonian is as in \eqref{eq:ham}, except that $|\kk|^{-1}$ is replaced by $1$; the Coulomb repulsion is thus replaced by the delta function, and the corresponding pair potential is replaced by $U \sum_{i<j} \delta(x_i-x_j)$. In one dimension the delta function is a perfectly good potential of a Schr\"odinger operator.

In this case we can also consider the Pekar approximation, whereby only variational functions of the form $\Psi=\psi(z_1,\ldots,z_N) \cdot \Phi$ are allowed. Here $\Phi$ is a vector in Fock space and $z_j=(x_j,\sigma_j)\in\R\times\{1,\ldots,q\}$ is a space-spin coordinate for an electron.

After minimizing the energy with respect to $\Phi$, one obtains the $N$-particle Pekar-Tomasevich functional (with spin)
\begin{equation}
 \label{eq:pt1d}
\sum_{i=1}^N \int |\nabla_i \psi|^2 \,dZ + U \sum_{i<j} \int \delta(x_i-x_j) |\psi(Z)|^2 \,dZ - \alpha \int_{\R} \rho_\psi(x)^2 \,dx \,.
\end{equation}
Here $\int dZ=\sum_{\sigma_1,\ldots,\sigma_N} \int_\R\cdots\int_\R dx_1\cdots dx_N$, and the density $\rho_\psi$ is defined by
$$
\rho_\psi(x) = N \sum_{\sigma_1,\ldots,\sigma_N} \int_{\R^{N-1}} |\psi(x_1,\sigma_1,\ldots,x_{N-1},\sigma_{N-1},x,\sigma_N)|^2 \,dx_1\cdots dx_{N-1} \,.
$$
We denote by $E^{(N)}_U(\alpha,q)$ the infimum of \eqref{eq:pt1d} over all antisymmetric $q$-state functions $\psi$ with $\int |\psi|^2 dZ=1$. This minimization problem also makes sense for $U=\infty$, where any finite energy wave function $\psi(z_1,\ldots,z_N)$ must vanish if $x_i=x_j$ for any $i\neq j$. We shall prove two facts about this minimization problem.

\begin{theorem}\label{pol1d}
 If $U=0$ and $N/q$ is an integer, then
$$
E^{(N)}_0(\alpha,q)=(N/q) E^{(q)}_0(\alpha,q) \,.
$$
If $U=\infty$, then
$$
E^{(N)}_\infty(\alpha,q)= N E^{(1)}_0(\alpha,1) \,.
$$
\end{theorem}

The field can cause multi-particle binding. A corollary of our first result is that, in the absence of repulsion, the energy per particle in the $q$-on state is at least as low as in any other state. That is, for any $N$ (not necessarily divisible by $q$)
$$
N^{-1} E_0^{(N)}(\alpha,q) \geq q^{-1} E_0^{(q)}(\alpha,q) \,.
$$
To see this, consider the particle number $M=Nq$ and apply Theorem \ref{pol1d} to this case. As a variational candidate for $E_0^{(M)}(\alpha,q)$ we can take $q$ lowest energy $N$-particle states infinitely separated from each other. Then we have $E_0^{(M)}(\alpha,q)\leq q E_0^{(N)}(\alpha,q)$. On the other hand, by Theorem 4.1, $E_0^{(M)}(\alpha,q) = N E_0^{(q)}(\alpha,q)$, and this proves our assertion.
Thus the $q$-on plays a similar role to that of nickel-62 in the curve of nuclear binding energies.

When $U=\infty$, the situation is even more dramatic; there is no binding of any kind.

One may say that in one-dimension antisymmetry trumps the attraction caused by the field. (This is not true in higher dimensions.) Presumably there are finite critical values of $U$ such that $p$\emph{-ons} break apart into $r$\emph{-ons} with $p>r\geq 1$, but we are not able to prove this. There should also be a finite critical value of $U$ above which there is no binding of any kind.

We now turn to the proof of Theorem \ref{pol1d}. We observe that the energy of a $q$\emph{-on} can be computed explicitly, as follows.

\begin{lemma}\label{bosonenergy}
 If $N=q$, then
$$
E^{(q)}_0(\alpha,q) = - \alpha^2 q^3/12 \,.
$$
\end{lemma}

\begin{proof}[Proof of Lemma \ref{bosonenergy}]
 Whatever $\rho_\psi$ might be, the minimum kinetic energy is realized by the product function $\psi(x_1,\ldots,x_q)=\phi(x_1)\cdots\phi(x_q)$, where $\phi(x)=\sqrt{ \rho_\psi(x)/q}$. Thus \cite{HoHo}
\begin{equation}
 \label{eq:ho}
\sum_{i=1}^N \int |\nabla_i \psi|^2 \,dZ \geq \int_{\R} |\nabla \sqrt{\rho_\psi}|^2 \,dx \,.
\end{equation}
Because there are $q$ spin states, there is an antisymmetric spin function of $q$ variables with which this product function can be multiplied to yield a valid antisymmetric space-spin function. Equality in \eqref{eq:ho} is then achievable.

To evaluate $E^{(q)}_0(\alpha,q)$ we have to find
$$
E^{(q)}_0(\alpha,q) = \inf\left\{ \int_\R \left( q|\phi'|^2 - \alpha q^2 |\phi|^4 \right)\,dx  : \|\phi\|_2=1 \right\} \,.
$$
The function $\phi(x)= (\alpha q/4)^{1/2} (1/\cosh(\alpha q x/2))$ is easily seen to be a solution to the corresponding Euler-Lagrange equation and, indeed, one can prove that it is the unique solution of the above minimization problem (up to translations and a complex phase) \cite{Ke,LT}. This leads to the desired expression for the energy.
\end{proof}

We need a slightly unorthodox version of a Lieb-Thirring inequality, which has been used before in \cite{LidL}:

\begin{lemma}\label{lt}
Assume that $N/q$ is an integer and let $\psi$ be a normalized, antisymmetric $q$-state function. Then
\begin{equation}
 \label{eq:lt}
\sum_{i=1}^N \int |\nabla_i \psi|^2 \,dZ \geq  \frac{3}{N q^2} \left( \int_\R \rho_\psi(x)^2 \,dx \right)^2 \,.
\end{equation}
\end{lemma}

\begin{proof}[Proof of Lemma \ref{lt}]
 Let $V=-W$ be a negative potential in $L^2(\R)$ and denote the eigenvalues of the one-dimensional Schr\"odinger operator $-\frac{d^2}{dx^2}-W$ by $\lambda_1\leq\lambda_2\leq\ldots$. If there is only a finite number $M$ of negative eigenvalues, we set $\lambda_{M+1}=\lambda_{M+2}=\ldots=0$. By the variational principle (see, e.g., \cite[Thm. 12.5]{LiLo}) we have
$$
 \sum_{i=1}^N \int \left( |\nabla_i \psi|^2 - W(x_i) |\psi(Z)|^2 \right) \,dZ \geq q \sum_{j=1}^{N/q} \lambda_j \,.
$$
By H\"older's inequality and the sharp Lieb-Thirring inequality \cite{LT} for $3/2$-moments of the eigenvalues, we find
$$
\sum_{j=1}^{N/q} |\lambda_j| \leq \left(\frac Nq\right)^{1/3} \left( \sum_{j=1}^{\infty} |\lambda_j|^{3/2} \right)^{2/3}
\leq \left(\frac Nq\right)^{1/3} \left( \frac3{16} \int_\R W(x)^2 \,dx \right)^{2/3} \,.
$$
To summarize, we have shown that
$$
\sum_{i=1}^N \int |\nabla_i \psi|^2 \,dZ \geq \int_\R W(x) \rho_\psi(x) \,dx - N^{1/3} q^{2/3} \left( \frac3{16} \int_\R W(x)^2 \,dx \right)^{2/3}
$$
for any $0\leq W\in L^2(\R)$. By choosing $W= c\rho_\psi$ and optimizing over the constant $c$, we obtain the desired bound \eqref{eq:lt}.
\end{proof}

\begin{proof}[Proof of Theorem \ref{pol1d}] 
 \emph{The case $U=0$.} We substitute the bound \eqref{eq:lt} into the expression \eqref{eq:pt1d} for the energy. This lower bound only depends on the unknown quantity $I=\int \rho_\psi(x)^2 \,dx$. By minimizing this expression with respect to $I$ we arrive at the lower bound $E^{(N)}_0(\alpha,q)\geq -\alpha^2 N q^2/12$. According to Lemma \ref{bosonenergy} this coincides with $N E^{(q)}_0(\alpha,q)/q$.

To conclude the proof, we need an upper bound of the same kind. This is easily done by noting that we can make a state of $N/q$ widely separated $q$\emph{-ons}. In the limit that the separation goes to infinity we obtain the upper bound of $N/q$ times the energy of a single $q$\emph{-on}.

\emph{The case $U=\infty$.} This case is easy in view of what we just proved. The electrons, regardless of their spin, cannot get past each other, i.e., the $N$-particle wave function vanishes whenever $x_i=x_j$ for some $i\neq j$. The configuration space is thus decomposed into a union of simplices of which $S=\{x_1<x_2<\ldots<x_N\}$ is representative.

Given a normalized $q$-state wave function $\psi$ we define a normalized, antisymmetric $1$-state wave function $\tilde \psi$ as follows: For $x\in S$ we set
$$
\tilde\psi(x) := \left( \frac{1}{N!} \sum_{\sigma_1,\ldots,\sigma_N} \sum_\pi |\psi(z_{\pi(1)},\ldots,z_{\pi(N)})|^2 \right)^{1/2}
$$
and we extend $\tilde\psi$ \emph{antisymmetrically} to the other simplices. A similar construction is used in \cite{LiMa}. The crucial point is that if $\psi$ has finite kinetic energy and vanishes on the boundaries of the simplices, then $\tilde\psi$ has finite kinetic energy as well and vanishes on the boundaries of the simplices. Moreover, by the convexity inequality for gradients \cite[Thm. 7.8]{LiLo} we have
$$
\sum_{i=1}^N \int_{\R^N} |\nabla_i \tilde\psi|^2 \,dx \leq \sum_{i=1}^N \int |\nabla_i \psi|^2 \,dZ \,.
$$
On the other hand, $\rho_{\tilde\psi}=\rho_\psi$, and therefore the total energy of $\psi$ is bounded from below by that of $\tilde\psi$. Note that these two energies coincide if the original $\psi$ was an antisymmetric function of space times a symmetric function of spin. To summarize, we have shown that $E^{(N)}_\infty(\alpha,q) = E^{(N)}_\infty(\alpha,1)$. Note that in the $q=1$ case, the repulsion energy vanishes because the antisymmetry forces the wave function to vanish on the boundaries of the simplices. Thus 
$E^{(N)}_\infty(\alpha,1)=E^{(N)}_0(\alpha,1)$, and the conclusion follows from the first part of the theorem.
\end{proof}

\subsection*{Acknowledgments}
 Partial financial support from the U.S.~National Science Foundation through grants PHY-0965859 (E.L.) and PHY-0845292 (R.S.) are gratefully acknowledged. L.T. would like to thank the PIMS Institute, University of British Columbia, for their hospitality and support.


\bibliographystyle{amsalpha}

\begin{thebibliography}{31}

\bibitem{DePeVe} J. T. Devreese, F. M. Peeters, G. Verbist, \textit{Large bipolarons in two and three dimensions}. Phys. Rev. {\bf B 43} (1991), 2712--2720.
\bibitem{DoVa} M. Donsker, S. R. S. Varadhan, \textit{Asymptotics for the polaron}. Comm. Pure Appl. Math. \textbf{36} (1983), 505--528.
\bibitem{Fe} R. P. Feynman, \textit{Slow electrons in a polar crystal}. Phys. Rev. \textbf{97} (1955), 660--665.
\bibitem{FrLiSeTh0} R. L. Frank, E. H. Lieb, R. Seiringer, L. E. Thomas, \textit{Bi-polaron and N-polaron binding energies}. Phys. Rev. Lett. \textbf{104} (2010), 210402.
\bibitem{FrLiSeTh} R. L. Frank, E. H. Lieb, R. Seiringer, L. Thomas, \textit{Stability and absence of binding for multi-polaron systems}. Preprint (2010), arXiv:1004.4892.
\bibitem{Fr} H. Fr\"ohlich, \textit{Theory of electrical breakdown in ionic crystals}. Proc. R. Soc. Lond. A \textbf{160} (1937), 230--241.
\bibitem{Fr2} J. Fr\"ohlich, \textit{Existence of dressed one-electron states in a class of persistent models}. Fortschr. Phys. \textbf{22} (1974), 159--198.
\bibitem{GeLo} B. Gerlach, H. L\"owen, \textit{Analytical properties of polaron systems or: Do polaronic phase transitions exist or not?} Rev. Mod. Phys. \textbf{63} (1991), 63--90.
\bibitem{GrMo} M. Griesemer, J. Schach M\o ller, \textit{Bounds on the minimal energy of translation invariant $N$-polaron systems}.  Comm. Math. Phys.  \textbf{297}  (2010),  no. 1, 283--297.
\bibitem{Gr} E. P. Gross, \textit{Strong coupling polaron theory and translational invariance}. Ann. Phys. \textbf{99} (1976), 1--29.
\bibitem{Gu} M. Gurari, \textit{Self-energy of slow electrons in polar materials}. Phil. Mag. Ser. 7 \textbf{44}:350 (1953), 329--336.\bibitem{HoHo} M. Hoffmann-Ostenhof, T. Hoffmann-Ostenhof, \textit{Schr\"odinger inequalities and asymptotics behavior of the electron density of atoms and molecules}. Phys. Rev. A \textbf{16} (1977), 1782--1785.
\bibitem{Ke} J. B. Keller, \textit{Lower bounds and isoperimetric inequalities for eigenvalues of the Schr\"odinger equation}. J. Math. Phys. \textbf{2} (1961), no. 2, 262--266. 
\bibitem{LePi} T-D. Lee, D. Pines, \textit{The motion of slow electrons in polar crystals}. Phys. Rev. \textbf{88} (1952), 960--961.
\bibitem{LeLoPi} T. D. Lee, F. Low, D. Pines, \textit{The motion of slow electrons in a polar crystal}. Phys. Rev. \textbf{90} (1953), 297--302.
\bibitem{Li} E. H. Lieb, \textit{Existence and uniqueness of the minimizing solution of Choquard's nonlinear equation}.
Studies in Appl. Math. \textbf{57} (1976/77), no. 2, 93--105. 
\bibitem{LidL} E. H. Lieb, M. de Llano, \textit{Solitons and the delta function fermion gas in Hartree-Fock theory}. J. Math. Phys. \textbf{19} (1978), no. 4, 860--868.
\bibitem{LiLo} E. H. Lieb, M. Loss, \textit{Analysis. Second edition}. Graduate Studies in Mathematics \textbf{14}, American Mathematical Society, Providence, RI, 2001.
\bibitem{LiMa} E. H. Lieb, D. C. Mattis, \textit{Theory of ferromagnetism and the ordering of electronic energy levels}. Phys. Rev.  \textbf{125} (1962), 164--172.
\bibitem{LiSe} E. H. Lieb, R. Seiringer, \textit{The stability of matter in quantum mechanics}, Cambridge (2010).
\bibitem{LiTh} E. H. Lieb, L. E. Thomas, \textit{Exact ground state energy of the strong-coupling polaron}. Comm. Math. Phys. \textbf{183} (1997), no. 3, 511--519. Erratum: \textit{ibid.} \textbf{188} (1997),  no. 2, 499--500.
\bibitem{LT} E.~H.~Lieb, W.~Thirring, \textit{Inequalities for the
      moments of the eigenvalues of the Schr\"odinger Hamiltonian and
      their relation to Sobolev inequalities}. Studies in Mathematical
    Physics, 269--303. Princeton University Press, Princeton, NJ, 1976.
\bibitem{LiYa} E. H. Lieb, K. Yamazaki, \textit{Ground-state energy and effective mass of the polaron}. Phys. Rev. \textbf{111} (1958), 728--722.
\bibitem{Mi} S. J. Miyake, \textit{Strong coupling limit of the polaron ground state}. J. Phys. Soc. Jpn. \textbf{38} (1975), 181--182.
\bibitem{Pe} S. I. Pekar, \textit{Untersuchung \"uber die Elektronentheorie der Kristalle}, Berlin, Akad. Verlag (1954).
\bibitem{Ro} G. Roepstorff, \textit{Path integral approach to quantum phy\-sics}. Berlin-Heidelberg-New York, Springer, 1994.
\bibitem{fomin} M. A. Smondyrev, V.M. Fomin, {\it Pekar-Fr\"ohlich bipolarons}. In: {\it Polarons and applications}, Proceedings in Nonlinear Science, V.D. Lakhno, ed., Wiley (1994).
\bibitem{SmKoVePeDe} M. A. Smondyrev, E. A. Kochetov, G. Verbist, F. M. Peeters, J. T. Devreese, \textit{Equivalence of 3D bipolarons in a strong magnetic field to 1D bipolarons}, Europhys. Lett. \textbf{19} (1992), 519.
\bibitem{SVPD2} M. A. Smondyrev, A. A. Shanenko,  J. T. Devreese, \textit{Stability criterion for large bipolarons in a polaron-gas background}, Phys. Rev. B {\bf 63} (2000), 024302.
\bibitem{SVPD1} M. A. Smondyrev, G. Verbist, F. M. Peeters, J. T. Devreese, \textit{Stability of multipolaron matter}, Phys. Rev. B  {\bf 47} (1993),  2596--2601.
\bibitem{Sp} H. Spohn, \textit{The polaron functional integral}. In: Stochastic processes and their applications, Dordrecht-Boston-London, Kluwer, 1990.
\bibitem{VaPeSmDe} P. Vansant, F. M. Peeters, M. A. Smondyrev, J. T. Devreese, \textit{One-dimensional bipolaron in the strong-coupling limit}, Phys. Rev. B \textbf{50} (1994), 12524.
\bibitem{VSPD} G. Verbist,  M. A. Smondyrev, F. M. Peeters, J. T. Devreese,
\textit{Strong-coupling analysis of large bipolarons in two  and three  dimensions}, Phys. Rev. B {\bf 45} (1992), 5262--5269.

\end{thebibliography}


\end{document}